\documentclass[sigconf]{acmart}

\setcopyright{rightsretained} 

\usepackage{color}
\usepackage{algpseudocode}
\usepackage{algorithm}
\usepackage{subcaption}
\usepackage{xspace}
\usepackage{booktabs}
\usepackage{pbox}
\usepackage{centernot}

\DeclareMathOperator*{\argmax}{arg\,max}
\newtheorem{defn}{Definition}
\newtheorem{problem}{Problem}

\newcommand{\favg}{f_{\text{avg}}}
\newcommand{\dom}{\text{dom}}
\newcommand{\ex}[1]{\mathbb{E}\left[#1\right]}

\begin{document}
\title{Deterministic \& Adaptive Non-Submodular Maximization\\via the Primal Curvature}

\author{J. David Smith}
\affiliation{%
	\department{CISE Department}
	\institution{University of Florida}
    \city{Gainesville}
    \state{Florida}
    \postcode{32611}
}
\email{jdsmith@cise.ufl.edu}

\author{My T. Thai}
\affiliation{%
	\department{CISE Department}
	\institution{University of Florida}
    \city{Gainesville}
    \state{Florida}
    \postcode{32611}
}
\email{mythai@cise.ufl.edu}

\begin{abstract}
While greedy algorithms have long been observed to perform well on a wide variety of problems, up to now approximation ratios have only been known for their application to problems having \textit{submodular} objective functions $f$. Since many practical problems have non-submodular $f$, there is a critical need to devise new techniques to bound the performance of greedy algorithms in the case of non-submodularity.

Our primary contribution is the introduction of a novel technique for estimating the approximation ratio of the greedy algorithm for maximization of monotone non-decreasing functions based on the curvature of $f$ without relying on the submodularity constraint. We show that this technique reduces to the classical $(1 - 1/e)$ ratio for submodular functions. Furthermore, we develop an extension of this ratio to the adaptive greedy algorithm, which allows applications to non-submodular stochastic maximization problems. This notably extends support to applications modeling incomplete data with uncertainty.
\end{abstract}

\maketitle

\section{Introduction}
It is well-known that greedy approximation algorithms perform remarkably well, especially when the traditional ratio of $(1 - 1/e) \approx 0.63$ \cite{nemhauser_analysis_1978} for maximization of {\em submodular} objective functions is considered. Over the four decades since the proof of this ratio, the use of greedy approximations has become widespread due to several factors. First, many interesting problems satisfy the property of \textit{submodularity}, which states that the marginal gain of an element never increases. If this condition is satisfied, and the set of possible solutions can be phrased as a uniform matroid, then one of the highest general-purpose approximation ratios is available ``for free'' with the use of the greedy algorithm. Second, the greedy algorithm is exceptionally simple both to understand and to implement.

A concrete example of this is the \textit{Influence Maximization} problem, to which the greedy algorithm was applied with great success -- ultimately leading to an empirical demonstration that it performed near-optimally on real-world data \cite{li_why_2017}. Kempe et al. showed this problem to be submodular under a broad class of influence diffusion models known as \textit{Triggering Models} \cite{kempe_maximizing_2003}. This led to a number of techniques being developed to improve the efficiency of the sampling needed to construct the problem instance (see e.g. \cite{borgs_maximizing_2014,tang_influence_2015,nguyen_stopandstare_2016} and references therein) while maintaining a $(1 - 1/e - \epsilon)$ ratio as a result of the greedy algorithm. This line of work ultimately led to a $(1 - \epsilon)$-approximation by taking advantage the dramatic advances in sampling efficiency to construct an IP that can be solved in reasonable time \cite{li_why_2017}. In testing this method, it was found that greedy solutions performed near-optimally -- an unexpected result given the $1 - 1/e$ worst-case.

For non-submodular problems, no general approximation ratio for greedy algorithms is known. However, due to their simplicity they frequently see use as simple baselines for comparison. On the Robust Influence Maximization problem proposed by He \& Kempe, the simple greedy method was used in this manner \cite{he_robust_2016}. This problem consists of a non-submodular combination of Influence Maximization sub-problems and aims to address uncertainty in the diffusion model. Yet despite the non-submodularity of the problem, the greedy algorithm performed no worse than the bi-criteria approximation \cite{he_robust_2016}.

Another recent example of this phenomena is the socialbot reconnaissance attack studied by Li et al. \cite{li_privacy_2016}. They consider a minimization problem that seeks to answer how long a bot must operate to extract a certain level of sensitive information, and find that the objective function is (adaptive) submodular only in a scenario where users disregard network topology. In this scenario, the corresponding maximization problem, Max-Crawling, has a $1 - 1/e$ ratio due to the work of Golovin \& Krause \cite{golovin_adaptive_2011}. However, this constraint does not align with observed user behaviors. They give a model based on the work of Boshmaf et al. \cite{boshmaf_socialbot_2011}, who observed that the number of mutual friends with the bot strongly correlates with friending acceptance rate. Although this model is no longer adaptive submodular, the greedy algorithm still exhibited excellent performance. Thus we see that while submodularity is \textit{sufficient} to imply good performance, it is is not \textit{necessary} for the greedy algorithm to perform well.

This, in turn, leads us to ask: is there any tool to theoretically bound the performance of greedy maximization with non-submodularity? Unfortunately, this condition has seen little study. Wang et al. give a ratio for it in terms of the worst-case rate of change in marginal gain (the \textit{elemental curvature} $\alpha$) \cite{wang_approximation_2014}. This suffices to construct bounds for non-submodular greedy maximization, though for non-trivial problem sizes they quickly approach 0. We note, however, that the $\alpha$ ratio still encodes strong assumptions about the worst case: that the global maximum rate of change can occur an arbitrary number of times.

Motivated by the unlikeliness of this scenario, our proposed bound instead works with an estimate of how much change can occur during the $k$ steps taken by the greedy algorithm.

The remainder of this paper is arranged as follows: First, we briefly cover the preliminary material needed for the proofs and define the class of problems to which they apply (Sec. \ref{sec:prelim}). We next define the notion of curvature used and develop a proof of the ratio based on it, with an extension to adaptive greedy algorithms, and show it is equivalent to the traditional $1 - 1/e$ ratio for submodular objectives (Sec. \ref{sec:ratio-theory}), and conclude with a reflection on the contributions and a discussion of future work (Sec. \ref{sec:conclusion}). 

\textbf{Contributions.}\vspace{-.4em}
\begin{itemize}
\item A technique for estimating the approximation ratio of greedy maximization of non-submodular monotone non-decreasing objectives on uniform matroids.
\item An extension of this technique to adaptive greedy optimization, where future greedy steps depend on the success or failure of prior steps.
\end{itemize}

\vspace{-.5em}
\subsection{Background \& Related Work}\label{sec:prelim}
To understand both the state of the art and advancements of this work, we first briefly cover each constraint required by the classical $1 - 1/e$ ratio \cite{nemhauser_analysis_1978}.
\subsubsection{Constraints on the \texorpdfstring{$1 - 1/e$}{1 - 1/e} Ratio}~\\
\textit{Uniform Matroids.}
A matroid defines the notion of dependencies between elements of a set, and are denoted by $\mathcal{M} = (X, \mathcal{I})$. $\mathcal{I} \subseteq 2^X$ is the set of \textit{independent subsets} of the universe $X$.\footnote{For a complete treatment on matroids and associated theory, see Oxley \cite{oxley_matroid_1992}.} For our purposes, it will suffice to cover the semantic meaning of $k$-uniform matroids, which is codified as follows:
\begin{enumerate}
\item All subsets $S$ of a feasible solution $T$ must also be feasible solutions.
\item Every $T \subset X, |T| = k$ is a feasible solution and is maximal in the sense that no superset $T \subset R \subset X$ is feasible.
\end{enumerate}
For general matroids, there exists a $1/2$ ratio for greedy maximization of submodular functions due to Fisher et al. \cite{fisher_analysis_1978}. This is a special case of their $1 / (p + 1)$ ratio for the intersection of $p$ matroids.


\textit{Submodularity.}
The submodularity condition states that given any subsets $S \subset T$ of a universe $X$, the marginal gain of any $x \in X$ does not increase as the cardinality increases: 
\[
f(T \cup \{x\}) - f(T) \leq f(S \cup \{x\}) - f(S)
\]
This formally encodes the idea of diminishing returns. Leskovec et al. exploited this property to show a data-dependent bound in terms of the marginal gain of the top-$k$ un-selected elements \cite{leskovec_costeffective_2007}, which was  generalized to the adaptive case \cite{golovin_adaptive_2011}. 

To the best of our knowledge, the only generally applicable relaxation of this constraint is the work of Wang et al. \cite{wang_approximation_2014}, who define a ratio in terms the \textit{elemental curvature} of a function, which encodes the degree with which a function may break submodularity.

\subsubsection{Alternate Problems \& Algorithms}
The $1 - 1/e$ ratio has shown surprising generality, with proofs that it holds for maximization of sequence functions \cite{zhang_string_2016a} (and references) and adaptive stochastic maximization of functions that are submodular in expectation \cite{golovin_adaptive_2011}, among others. However, not all adjacent work relies on the same na\"{i}ve greedy method. To obtain a bound on the relaxation of monotonicity, Buchbinder et al. \cite{buchbinder_tight_2012} proposed a ``double-greedy'' algorithm with a $1/3$ (deterministic) or $1/2$ (randomized) ratio. For maximization on an intersection of $p \geq 2$ matroids, Lee et al. showed a $1 / (p + \epsilon)$, $\epsilon > 0$ ratio for a local search method \cite{lee_submodular_2010}.

Vondrak et al. proposed a continuous greedy algorithm with a $(1/c) (1 - e^{-c})$ ratio for general matroids \cite{vondrak_submodularity_2010}, where $c$ is the total curvature of the function. An augmentation of this method has been shown to obtain a $(1 - c/e)$-approximation for single matroids \cite{sviridenko_optimal_2015}, along with an analogue for supermodular minimization. We remark that, while it exhibits a better ratio, this comes with a corresponding increase in complexity of the algorithm. 

\subsubsection{Curvature-Based Ratios}
Conforti \& Cornu\'{e}jols \cite{conforti_submodular_1984} introduced the idea of \textit{total curvature} later used by Sviridenko et al. for their $(1 - c/e)$ ratio.
\begin{defn}[Total Curvature]
Given a monotone non-decreasing submodular function $f$ defined on a matroid $\mathcal{M} = (\mathcal{I}, X)$, the total curvature of $f$ is
\[
c = \max_{j \in X} \left\{ 1 - \frac{f(X) - f(X \setminus \{j\})}{f(\{j\}) - f(\emptyset)}\right\}
\]
\end{defn}

Using this definition, they arrived at a $1 / (1 + c)$ approximation for general matroids, which reduces to $\frac{1}{c}(1 - e^{-c})$ for maximzation on uniform matroids. Recently, Wang et al. \cite{wang_approximation_2014} extended this idea by introducing the \textit{elemental curvature} $\alpha$ of a function $f$:
\begin{defn}[Elemental Curvature]
  The elemental curvature of a monotone non-decreasing function $f$ is defined as
  $$ \alpha = \max_{S \subseteq X, i, j \in X} \frac{f_i(S \cup \{j\})}{f_i(S)}$$

  where $f_i(S) = f(S\cup \{i\}) - f(S)$.
\end{defn}
While the resulting ratio (Theorem \ref{thm:wang-ele-ratio}) is not as clean as that of prior work, this ratio is well-defined for non-submodular functions.

\begin{theorem}[Wang et al. \cite{wang_approximation_2014}]\label{thm:wang-ele-ratio}
For a monotone non-decreasing function $f$ defined on a $k$-uniform matroid $\mathcal{M}$, the greedy algorithm on $\mathcal{M}$ maximizing $f$ produces a solution satisfying \[\left[1 - \left(1 - A_k^{-1}\right)^k\right] f(S^*) \leq f(S) \]
where $S$ is the greedy solution, $S^*$ is the optimal solution, $A_k = \sum_{i=1}^{k-1} \alpha^i$ and $\alpha$ is the elemental curvature of $f$.
\end{theorem}

\begin{corollary}[Wang et al. \cite{wang_approximation_2014}]
When $\alpha = 1$, the ratio given by Theorem \ref{thm:wang-ele-ratio} converges to $1 - 1/e$ as $k \rightarrow \infty$.
\end{corollary}

However, the ratios produced based on the elemental curvature rapidly converge to $0$ for non-submodular functions. This behavior is shown in Figure \ref{fig:alpha-ratios}. Even for $k = 25$, the ratio is effectively zero and therefore uninformative. In contrast, we show that our ratio produces significant bounds for two non-submodular functions, while still converging to the $1 - 1/e$ ratio for submodular functions.

\begin{figure}[thb]
\begin{subfigure}{0.2\textwidth}
\includegraphics[width=0.98\textwidth]{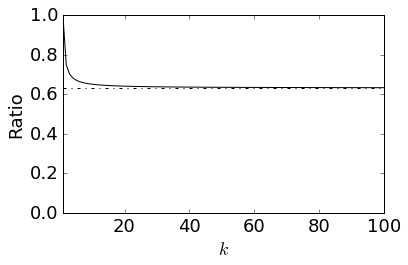}
\caption{\label{subfig:alpha-submod}$\alpha = 1.0$}
\end{subfigure}
\begin{subfigure}{0.2\textwidth}
\includegraphics[width=0.98\textwidth]{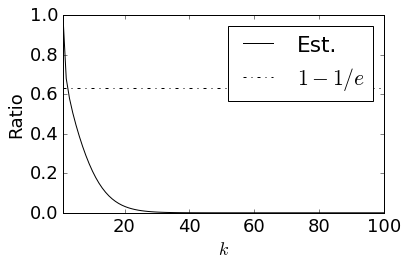}
\caption{\label{subfig:alpha-supmod}$\alpha = 1.3$}
\end{subfigure}
\caption{\label{fig:alpha-ratios} The ratio produced by Theorem \ref{thm:wang-ele-ratio} for (\subref{subfig:alpha-submod}) submodular and (\subref{subfig:alpha-supmod}) non-submodular functions.}
\vspace{-.5em}
\end{figure}

\section{A Ratio for \texorpdfstring{\MakeLowercase{$f$} }{}Non-Submodular\texorpdfstring{}{ Objectives}}\label{sec:ratio-theory}
In this section, we introduce a further extension to the notion of curvature: primal curvature. We derive a bound based on this, prove its equivalence to $1 - 1/e$ for submodular functions. Then, we extend the ratio to the adaptive case, which allows direct application to a number of problems modeled under incomplete knowledge. We adopt a problem definition similar to that of Wang et al. Specifically, our ratio applies to any problem that can be phrased as $k$-Uniform Matroid Maximization.
\begin{problem}[$k$-Uniform Matroid Maximization]
Given a $k$-uniform matroid $\mathcal{M} = (X, \mathcal{I})$ and a monotone non-decreasing function $f: 2^X \rightarrow \mathbb{R}$, find
\[
S = \argmax_{I \in \mathcal{I}} f(I)
\]
\end{problem}

\subsection{Construction of the Ratio}
As noted previously, the ratio given by elemental curvature rapidly converges to zero for non-submodular functions. We observe that this is due to the definition of $\alpha$ encoding the worst-case potential, and address this limitation by introducing the \textit{primal curvature} of a function. Our definition separates the notion of rate-of-change from the global perspective imposed by elemental curvature.\footnote{The term \textit{primal} is adopted primarily to distinguish this definition from prior work.}

\begin{defn}[Primal Curvature]
  The primal curvature of a set function $f$ is defined as
  \[ \nabla_f(i, j \mid S) = \frac{f_i(S \cup \{j\})}{f_i(S)}\]
  
  The global maximum primal curvature is equivalent to the elemental curvature of a function.
\end{defn}

This shift from global to local perspective allows focus on the patterns present in real-world problem instances rather than limiting our attention to the worst-case scenarios.

A key observation of Wang et al's work is that the elemental curvature defines an upper bound on the change between $f(S)$ and $f(T)$, for some $S \subset T$, in terms of $\alpha$ and the marginal gain at $S$. The definition of primal curvature improves on this, giving an equivalence in terms of the \textit{total primal curvature} $\Gamma$.

\begin{defn}[Total Primal Curvature]\label{defn:tpc}
    The total primal curvature of \(x \in X\) between two sets \(S \subseteq T \subset X\)
    with \(x \not\in T\) is 
    \[\Gamma (x \mid T, S) = \prod_{j=1}^r
        \nabla_f (x, t_j \mid S \cup \{t_1, t_2, \ldots, t_{j-1}\})\]
    where the \(t_j\)'s form an arbitrary ordering of
    \(T \setminus S\) and \(r = |T \setminus S|\).  
\end{defn}

We note that $\Gamma$ can be interpreted as the \textit{total change} in the marginal value of $x$ from point $A$ to point $B$.  The following lemma illustrates this, as well as providing a useful identity.

\begin{lemma}\label{lemma:tpc-ident}
	\[\Gamma(x \mid T, S) = \frac{f_x(S \cup T)}{f_x(S)}\]
\end{lemma}
\begin{proof}
First, expand the product into its constituent terms:
\[\frac{f_x(S \cup \{t_1\})}{f_x(S)} \cdot \frac{f_x(S \cup \{t_1, t_2\})}{f_x(S \cup \{t_1\})} \cdots \frac{f_x(S \cup T)}{f_x(S \cup \{t_1, t_2, \cdots t_{r-1}\})} \]
After cancelling, the statement immediately follows.
\end{proof}

From this identity, we gain one further insight: the order in which elements are considered in $\Gamma$ does not matter.

\begin{corollary}\label{tpc-ordering}
The product $\Gamma(x \mid T, S)$ is order-independent.
\end{corollary}

Using this, we can prove an equivalence between the change in total benefit and the sum of marginal gains taken with respect to $S$.
\begin{lemma}\label{lemma:tpc-equiv}
For a set function $f$ and a pair of sets $S \subseteq T$,
\[f(T) - f(S) = \sum_{j=1}^r \Gamma(t_j \mid S_{j-1}, S) f_{t_j}(S)\]
where $r = |T \setminus S|$, $f_{x}(S) = f(S \cup \{x\}) - f(S)$ is the marginal gain and $S_{j-1} = S \cup \{j_1, j_2, \ldots j_{i-1}\}$.
\end{lemma}

\begin{proof}
Let $j_1$ be an arbitrary labeling of $T \setminus S$. Then we have:

\[
f(T) - f(S) = f(S \cup \{j_1, j_2, \ldots j_r\}) - f(S)
= \sum_{t=1}^r f_{j_t}(S_{t-1})\]

By the identity given in Lemma \ref{lemma:tpc-ident}, we can write
\[f(T) - f(S) = \sum_{t=1}^r \Gamma(j_t \mid S_{t-1}, S) f_{j_t}(S)\]

Noting that $S \cup S_{i} = S_{i}$. Thus, the statement is proven.
\end{proof}

With this lemma, we can now construct the ratio.

\begin{theorem}\label{thm:tpc-ratio}
For a monotone non-decreasing function $f: 2^X \rightarrow \mathbb{R}$, the greedy algorithm on a $k$-uniform matroid $\mathcal{M} = (X, \mathcal{I})$ maximizing $f$ produces a solution satisfying
\begin{equation}\label{eqn:ratio-final}
\left[1 + \left(\frac{f(S^+)}{f(S)} - 1\right)\hat\Gamma(S)\right]^{-1} f(S^*) \leq f(S) 
\end{equation}
where $S$ is the greedy solution, $S^+ = S \cup \{g_{k+1}\}$ is the greedy solution for an identical problem if a $k+1$-uniform supermatroid $\mathcal{M}^+$ of $\mathcal{M}$ is well-defined, $S^*$ is the optimal solution on $\mathcal{M}$, and $\hat\Gamma(S)$ is an estimator satisfying:
\[\forall T \in \mathcal{I}:\sum_{j_t \in T \setminus S} \Gamma(j_t \mid S_{t-1}, S) \leq \hat\Gamma(S) \]
where $S_{t-1} = S \cup \{j_1, j_2, \ldots, j_{t-1}\}$
\end{theorem}
\begin{proof}
To begin, note that $f(S^*) \leq f(S^* \cup S)$ due to $f$ monotone non-decreasing. Then, by Lemma \ref{lemma:tpc-equiv} we have:
\begin{equation}\label{eqn:ratio-start}f(S^* \cup S) - f(S) = \sum_{t=1}^r \Gamma(x \mid S_{t-1}, S) f_{j_t}(S)\end{equation}

We observe that any ratio that requires knowing $S^*$ is of little practical value: if $S^*$ is known, we can simply compute $f(S)/f(S^*)$. Therefore, we relax our assumptions in three key ways to go from Eqn. \eqref{eqn:ratio-start}, which assumes that we know $S^*$ exactly, to Eqn. \eqref{eqn:ratio-final}, which requires no knowledge of the optimal.

First, we partly remove the assumption on knowledge of $j_t \in S^*$ by substituting $f_{j_t}(S)$ with $f_{g_{k+1}}(S)$, where $g_{k+1} = \argmax_{x} f_x(S)$.
\[f(S^*) - f(S) \leq f(S^* \cup S) - f(S) \leq f_{g_{k+1}}(S) \sum_{t=1}^r \Gamma(j_t \mid S_{t-1}, S)\]

Next, we apply the upper bound $\hat\Gamma$ as defined above to both remove the remaining dependence on knowledge of $j_t$ and to eliminate the requirement of knowing $|S^* \setminus S|$.
\begin{equation}\label{eqn:gamma-relation}
f(S^*) - f(S) \leq f_{g_{k+1}}(S) \hat\Gamma(S)
\end{equation}
Then, rearranging terms we get
\begin{align*}
f(S^*) &\leq f(S) + f_{g_{k+1}}(S)\hat\Gamma(S)\\
&= f(S) + \left(f(S^+) - f(S)\right)\hat\Gamma(S)
\end{align*}
where $S^+ = S \cup \{g_{k+1}\}$. Then, dividing through by $f(S)$ and cross-multiplying, we get:
\begin{equation*}
\left[1 + \left(\frac{f(S^+)}{f(S)} - 1\right)\hat\Gamma(S) \right]^{-1} f(S^*) \leq f(S)
\end{equation*}
\end{proof}

When compared to traditional approximation ratios, this ratio has several obvious differences. First, it has dependencies on both the greedy solution and an extension of it to $k+1$ elements. This is both a strength and fundamental limitation of Theorem \ref{thm:tpc-ratio}: it takes into account how much the greedy solution has converged toward negligible marginal gains, but also inhibits general analysis over all potential problem instances. Further, it requires that the supermatroid $\mathcal{M}^+$ be well-defined, though we remark that this is generally not a problem. In practice, most problems solved with greedy algorithms are $k$-element solutions on $n$-element spaces, with $k$ typically much less than $n$.

\subsection{Equivalence to the \texorpdfstring{$1 - 1/e$}{1 - 1/e} Ratio}

We next show that under assumptions encoding the submodularity condition, the above is equivalent to the $1 - 1/e$ ratio as $k \rightarrow \infty$.

\begin{lemma}\label{lemma:fixed-gamma-bound}
    Given a $\hat\Gamma$ satisfying \(\forall G: \hat\Gamma \geq \hat\Gamma(G)\), the
    greedy algorithm produces a \(k\)-element solution \(S\)
    satisfying
    \[\left[1 - \left(1 - \hat\Gamma^{-1}\right)^k\right] f(S^*) \leq
        f(S) \]
\end{lemma}
\begin{proof}
    We begin with Eqn. (\ref{eqn:gamma-relation}):
    \[
	f(S^*) - f(S_l) \leq f_{g_{l+1}}(S) \hat\Gamma(S_l)
	\]
    for each $l \leq k$, where $S_l$ denotes the $l$-element greedy solution. Substitute $\hat\Gamma$ for $\hat\Gamma(S_l)$. 
    Multiplying both sides by \((1 - \hat\Gamma^{-1})^{k - l}\) and summing
    from \(l = 1\) to \(l = k\). The left-hand side becomes:
    \[\hat\Gamma\left[1 - \left(\frac{\hat\Gamma -
                    1}{\hat\Gamma}\right)^k\right]f(S^*) = \hat\Gamma\left[1
            - \left(1 - \hat\Gamma^{-1}\right)^k\right]f(S^*)\]
    
    To obtain the right-hand side, separate \(f(S_l) = \sum_{i=1}^l
        f_{g_{i}}(S_{i-1})\) into the marginal gain terms to produce the
    following in the body of the summation:
    \[\left(\hat\Gamma(1 - \hat\Gamma^{-1})^{k-l} + \sum_{i={l+1}}^k(1 -
            \hat\Gamma^{-1})^{k-i}\right)f_{g_{l+1}}(S_l)\]
    Summing this over $l$ and employing the identity of the geometric
    series, this reduces to \(\hat\Gamma f(S_k) = \hat\Gamma f(S)\) on the right-hand side.
    Thus, we obtain the relation 
    \[\left[1 - \left(1 - \hat\Gamma^{-1}\right)^k\right] f(S^*) \leq
        f(S) \]
\end{proof}

\begin{corollary}\label{cor:old-bound}
    For a submodular monotone non-decreasing function $f$, the following
    relation holds as \(k \rightarrow \infty\):
    \[(1 - 1/e)f(S^*) \leq f(S)\]
\end{corollary}
\begin{proof}
    For a submodular function, the primal curvature of any two elements
    \(u, v\) at any point \(T\) satisfies \(\nabla(u, v \mid T) \leq 1\)
    by the definition of submodularity. Thus, we obtain directly that $\hat\Gamma = k$ satisfies the requisite relation. Then, the
    limit of \((1 - \hat\Gamma^{-1})^k = (1 -
        \frac{1}{k})^k\) as \(k \rightarrow \infty\) is \(1/e\), leading directly to the statement above.
\end{proof}

Thus, we see that this ratio is a generalization of the classical $1 - 1/e$ approximation ratio that allows \textit{specialization} of a ratio to the particular kind of problem instances being operated on. Further, the definition of total primal curvature illuminates why this ratio is capable of producing more useful bounds for non-submodular objectives than that of Wang et al: the $\Gamma$ values encode a product of values that \textit{may} converge to a limit, depending on problem instance, while the $\alpha$ bound uses $\prod_{t=0}^i \alpha = \alpha^i$ which does not converge for any $\alpha > 1$ (a condition which is implied by non-submodularity).

\subsection{The Adaptive Ratio}
We conclude this section by extending this ratio to the adaptive case where the decision made at each greedy step takes into account the outcomes of previous decisions. Briefly: in an adaptive algorithm, at each step the algorithm has a \textit{partial realization} $\psi$ consistent with the true realization $\Phi$ \cite{golovin_adaptive_2011}. After each step, this partial realization is updated with the outcome of that step to form $\psi'$. The method for deciding the steps to take is termed a \textit{policy}, with the greedy algorithm encoded as the greedy policy. 

This representation supports the study of algorithms that operate with incomplete information and gradual revelation of the data. The initial motivation was described in terms of placement of sensors that may fail, and this technique has seen further use in studying networks with incomplete topology \cite{li_privacy_2016,seeman_adaptive_2013}, active learning under noise \cite{golovin_nearoptimal_2010}, and distributed representative subset mining \cite{mirzasoleiman_distributed_2013}.


We generalize our ratio to this case by defining the \textit{adaptive} primal curvature of a function in terms of the partial realizations.

\begin{defn}[Adaptive Primal Curvature]
The primal curvature of an adaptive monotone non-decreasing function $f$ is \[
\nabla_f(i, j \mid \psi) = \mathbb{E}\left[\frac{\Delta(i \mid \psi \cup s)}{\Delta(i \mid \psi)} \;\middle|\; s \in S(j)\right]
\]
where $S(j)$ is the set of possible states of $j$ and $\Delta$ is the conditional expected marginal gain \cite{golovin_adaptive_2011}.
\end{defn}

\begin{defn}[Adaptive T.P.C.]
Let $\psi \subset \psi'$ and $\psi \rightarrow \psi'$ represent the set of possible state sequences leading from $\psi$ to $\psi'$. Then the adaptive total primal curvature is\[
\Gamma(i\mid \psi', \psi) = \mathbb{E}\left[\prod_{s_j \in Q} \nabla'(i, s_j \mid \psi \cup \{s_1, \ldots, s_{j-1}\}) \;\middle|\; Q \in \psi \rightarrow \psi'\right]
\]
\end{defn}

This definition leads to the following theorem by similar arguments as Thm. \ref{thm:tpc-ratio}. However, the operations within expectation require additional care.

\begin{lemma}\label{lemma:adaptive-tpc-ident}
\[\Gamma(i \mid \psi', \psi) = \frac{\Delta(i \mid \psi')}{\Delta(i \mid \psi)}\]
\end{lemma}
\begin{proof}
Fix a sequence $Q \in \psi \rightarrow \psi'$ of length $r$. Then, expanding the product we obtain
\[
\frac{\Delta(i \mid \psi \cup \{s_1\})}{\Delta(i \mid \psi)} \cdot \frac{\Delta(i \mid \psi \cup \{s_1, s_2\})}{\Delta(i \mid \psi \cup \{s_1\})} \cdots \frac{\Delta(i \mid \psi')}{\Delta(i \mid \psi' \setminus \{s_{r-1}\})}
\]
If we take the expectation of this w.r.t. the possible sequences $Q$, we obtain the same ratio regardless of $Q$, and therefore the claim holds trivially.
\end{proof}

\begin{corollary}\label{cor:adaptive-gamma-bound}
Suppose that $\forall \psi' \supset \psi, i \not\in \dom(\psi'): \Gamma(i \mid \psi', \psi) \leq \hat\Gamma(\psi)$. Then \[
\Delta(i \mid \psi') \leq \hat\Gamma(\psi) \Delta(g_{l+1}\mid \psi)
\]
where $\psi$ is the partial realization resulting from application of the $l$-element greedy policy, $\psi \subset \psi'$, $i \not\in \dom(\psi')$, and $g_{l+1}$ is the next element that would be selected by the greedy policy.
\end{corollary}
\begin{proof}
By Lemma \ref{lemma:adaptive-tpc-ident}, \[
\Delta(i \mid \psi') = \Gamma(i \mid \psi', \psi) \Delta(i \mid \psi) \leq \hat\Gamma(\psi) \Delta(g_{l+1} \mid \psi)
\]
and thus the statement holds.
\end{proof}

\begin{lemma}\label{lemma:adaptive-opt-relation}
\begin{equation}
\favg(\pi') - \favg(\pi_l) \leq k\hat\Gamma(\pi_l)\Delta_{avg}(\pi_l, \pi_{l+1})
\end{equation}
where $\pi_l$ is the $l$-truncation of $\pi$ with $l < k$, $\pi'$ selects exactly $k$ elements, $\hat\Gamma(\pi_l) = \max_{\psi = \pi_l(\mathbf{\Phi})} \hat\Gamma(\psi)$ is the maximum over all possible realizations resulting from applying policy $\pi_l$, and $\Delta_{avg}(\pi_l, \pi_{l+1}) = \favg(\pi_{l+1}) - \favg(\pi_l)$.
\end{lemma}
\begin{proof}
By Corollary \ref{cor:adaptive-gamma-bound}, we have \begin{align*}
&\favg(\pi') - \favg(\pi_l) \leq \ex{k\hat\Gamma(\psi)\Delta(g_{l+1}\mid \psi)\mid \psi}\\
&= k\hat\Gamma(\pi_l)\ex{\Delta(g_{l+1} \mid \psi) \mid \psi}\\
&= k\hat\Gamma(\pi_l)\ex{\ex{f(\dom(\psi) + g_{l+1}, \Phi) - f(\dom(\psi), \Phi) \mid \Phi \sim \psi}\mid \psi}\\
&= k\hat\Gamma(\pi_l)\ex{f(E(\pi_{l+1}, \Phi), \Phi) - f(E(\pi_l, \Phi), \Phi) \mid \Phi}\\
&= k\hat\Gamma(\pi_l)\Delta_{avg}(\pi_l, \pi_{l+1})
\end{align*}
where the first equality uses the definition $\hat\Gamma(\psi) \leq \hat\Gamma(\pi_l)$ and the second uses the definition of $\Delta(\cdot)$.
\end{proof}

\begin{theorem}
Define $\hat\Gamma_k(\pi) = \max_{0 \leq l \leq k} \hat\Gamma(\pi_l)$. Then \begin{equation}
\left[1 - \left(1 - \frac{1}{k\hat\Gamma_k(\pi)}\right)^k\right]\favg(\pi^*_k) \leq \favg(\pi_k)
\end{equation}
\end{theorem}
\begin{proof}
By Lemma \ref{lemma:adaptive-opt-relation}, we have \[ 
\favg(\pi^*_k) \leq \favg(\pi_l) + k\hat\Gamma(\pi_l)\Delta_{avg}(\pi_l, \pi_{l+1})
\]
Multiply both sides by $(1 - (k\hat\Gamma_k(\pi))^{-1})^{k-1-l}$ and sum from $l=0$ to $k-1$. We get that the left hand side reduces to \[
k\hat\Gamma_k(\pi)\left[1 - \left(\frac{k\hat\Gamma_k(\pi) - 1}{k\hat\Gamma_k(\pi)}\right)^k\right]\favg(\pi^*_k)
\]
and the right hand side reduces to $k\hat\Gamma_k(\pi)\favg(\pi_k)$ by employing the identity for partial sums of a geometric series to find that each term of the outer sum has coefficient $k\hat\Gamma_k(\pi)$. Combining these, we directly obtain the statement of the theorem.
\end{proof}
\vspace{-.3em}\section{Conclusion \& Future Work}\label{sec:conclusion}
In this paper, we presented a method for estimating the approximation ratio of greedy maximization that works transparently for both submodular and non-submodular functions, in addition to a variant supporting adaptive greedy algorithms. This ratio reduces to at worst $1 - 1/e$ as $k \rightarrow \infty$ for submodular functions, and is shown to provide performance bounds for non-submodular maximization.

While we have demonstrated the utility of our technique for understanding the performance of non-submodular maximization, there remains room for further development. Relaxations of the uniformity and monotonicity conditions have found widespread use for submodular functions, and we expect that relaxing them for this ratio would likewise be generally useful.
\vspace{-0.5em}
\bibliographystyle{ACM-Reference-Format}
\bibliography{kdd17}
\end{document}